\documentclass[11pt]{llncs}
\usepackage{fullpage}

\newcommand{\Oh}[1]
    {\ensuremath{\mathcal{O}\left( {#1} \right)}}
\newcommand{\Om}[1]
    {\ensuremath{\Omega\left( {#1} \right)}}

\begin{document}

\title{A Lower Bound on the Complexity of\\
    Approximating the Entropy of a Markov Source}
\author{Travis Gagie}
\institute{Department of Computer Science\\
    University of Chile\\
    \email{travis.gagie@gmail.com}}
\maketitle

The Asymptotic Equipartition Property (see, e.g.,~\cite{CT06}) implies that, if we choose the characters of a string $s$ of length $n$ independently and according to the same probability distribution $P$ over the alphabet then, for large values of $n$, the 0th-order empirical entropy \(H_0 (s)\) of $s$ (see, e.g.,~\cite{Man01}) will almost certainly be close to the entropy \(H (P)\) of $P$.  Batu, Dasgupta, Kumar and Rubinfeld~\cite{BDKR05} showed that, if \(H (P) = \Om{\gamma / \epsilon}\), then we can almost certainly approximate \(H (P)\) to within a factor of $\gamma$ after seeing $\Oh{\sigma^{(1 + \epsilon) / \gamma^2} \log \sigma}$ characters of $s$, where $\sigma$ is the alphabet size and $\epsilon$ is any positive constant; they proved a lower bound of $\Om{\sigma^{1 / (2 \gamma^2)}}$, which was later improved by Raskhodnikova, Ron, Shpilka and Smith~\cite{RRSS07} and Valiant~\cite{Val08}.

Similarly, the Shannon-McMillan-Breiman Theorem (see, e.g.,~\cite{CT06} again) implies that, if we generate $s$ from a stationary ergodic $k$th-order Markov source $\mathcal{X}$ then, for large values of $n$, the $k$th-order empirical entropy \(H_k (s)\) of $s$ (see, e.g.,~\cite{Man01} again) will almost certainly be close to the entropy \(H (\mathcal{X})\) of $\mathcal{X}$.  Although many papers have been written about approximating the entropy of a Markov source based on a sample (see, e.g.,~\cite{CKV04} and references therein), we know of no upper or lower bounds similar to Batu et al.'s results.  We now give a simple proof that, even if we know $\mathcal{X}$ has entropy either $0$ or at least \(\log (\sigma - k)\), there is still no algorithm that, with probability bounded away from \(1 / 2\), guesses its entropy correctly after seeing at most \((\sigma - k)^{k / 2 - \epsilon}\) characters.

\begin{lemma}
For any \(k \geq 1\), \(\epsilon > 0\) and sufficiently large $\sigma$, there is a $k$th-order Markov source over the alphabet \(\{0, \ldots, \sigma - 1\}\) that has entropy at least \(\log (\sigma - k)\) but, with high probability, does not emit duplicate $k$-tuples among its first \((\sigma - k)^{k / 2 - \epsilon}\) characters.
\end{lemma}

\begin{proof}
Consider the $k$th-order Markov source that, whenever it has emitted a $k$-tuple \(\alpha = a_1, \ldots, a_k\), emits a character drawn uniformly at random from \(\{0, \ldots, \sigma - 1\} - \{a_1, \ldots, a_k\}\).  Notice this source has entropy at least \(\log (\sigma - k)\).  Also, a $k$-tuple $\alpha$ cannot occur in position $i$ if it occurs in any of the positions \(i - k + 1, \ldots, i - 1, i + 1, \ldots, i + k - 1\), and vice versa.  Finally, the probability $\alpha$ occurs in position $i$ is independent of whether it occurs in position $j$ for \(j \leq i - k\) or \(j \geq i + k\).

For \(i - k + 1 \leq j \leq i + k - 1\), let the indicator variable $B_j$ be $1$ if $\alpha$ occurs in position $j$, and $0$ otherwise.  By Bayes' Rule, the probability $\alpha$ occurs in position $i$, given that it does not occur in any of the positions \(i - k + 1, \ldots, i - 1, i + 1, \ldots, i + k - 1\), is
\begin{eqnarray*}
\lefteqn{\Pr \left[ B_i = 1 \,\left|\, \rule{0ex}{2ex} \right.
    B_{i - k + 1} = \cdots = B_{i - 1} = B_{i + 1} = \cdots = B_{i + k - 1} = 0 \right]}\\[1ex]
& = & \frac{\Pr \left[ \rule{0ex}{2ex} B_i = 1\ \mbox{and}\
    B_{i - k + 1} = \cdots = B_{i - 1} = B_{i + 1} = \cdots = B_{i + k - 1} = 0 \right]}
    {\Pr \left[ \rule{0ex}{2ex}
    B_{i - k + 1} = \cdots = B_{i - 1} = B_{i + 1} = \cdots = B_{i + k - 1} = 0 \right]}\\[1ex]
& \leq & \frac{\Pr [B_i = 1]}
    {1 - \Pr \left[ \rule{0ex}{2ex} B_{i - k + 1} = 1\ \mbox{or}\ \cdots\ \mbox{or}\ B_{i - 1} = 1\ \mbox{or}\ B_{i + 1} = 1\ \mbox{or}\ \cdots\ \mbox{or}\ B_{i + k - 1} = 1 \right]}\\[1ex]
& \leq & \frac{1 / (\sigma - k)^k}{1 - (2 k - 2) / (\sigma - k)^k}\\[1ex]
& = & \frac{1}{(\sigma - k)^k - 2 k - 2}\,.
\end{eqnarray*}
It follows that the probability $\alpha$ occurs at least twice among the first \((\sigma - k)^{k / 2 - \epsilon}\) emitted characters is at most the probability that, while drawing \((\sigma - k)^{k / 2 - \epsilon}\) elements uniformly at random and with replacement from a set of size \((\sigma - k)^k\), we draw a specified element at least twice.  Therefore, the probability any $k$-tuple occurs at least twice among the first \((\sigma - k)^{k / 2 - \epsilon}\) emitted characters is at most the probability that we draw any element at least twice.  For \(k \geq 1\) and sufficiently large $\sigma$, both probabilities are negligible. \qed
\end{proof}

\begin{theorem}
Suppose that, for any \(k \geq 1\), \(\epsilon > 0\) and sufficiently large $\sigma$, we are given a black box that allows us to sample characters from a $k$th-order Markov source over the alphabet \(\{0, \ldots, \sigma - 1\}\).  Even if we know the source has entropy either $0$ or at least \(\log (\sigma - k)\), there is still no algorithm that, with probability bounded away from \(1 / 2\), guesses the entropy correctly after sampling at most \((\sigma - k)^{k / 2 - \epsilon}\) characters.
\end{theorem}

\begin{proof}
Consider any algorithm $A$ for guessing the source's entropy.  Suppose there is a string $s$ of length \((\sigma - k)^{k / 2 - \epsilon}\) containing no duplicate $k$-tuples and such that, with probability at least \(1 / 2\), $A$ stops and guesses ``at least \(\log (\sigma - k)\)'' after sampling a prefix of $s$.  Then on any source with entropy $0$ that starts by emitting $s$ with probability $1$ the algorithm errs with probability at least \(1 / 2\).  Given $s$, it is straightforward to build such a source.

Now suppose there is no such string $s$.  Then whenever the first \((\sigma - k)^{k / 2 - \epsilon}\) sampled characters contain no duplicate $k$-tuples, $A$ either samples more characters or stops and guesses ``0'', with probability at least \(1 / 2\).  Therefore, on any source with entropy at least \(\log (\sigma - k)\) that, with high probability, does not emit duplicate $k$-tuples among its first \((\sigma - k)^{k / 2 - \epsilon}\) characters --- such as the one described in the lemma above --- $A$ either samples more characters or errs, with probability nearly \(1 / 2\).  \qed
\end{proof}

\bibliographystyle{plain}
\bibliography{entropy}

\end{document}